\documentclass[11pt,reqno]{amsart}
\usepackage{hyperref}
\usepackage{latexsym}
\usepackage{amsmath}
\usepackage{enumerate}
\usepackage{amsfonts}
\usepackage{amssymb}
\usepackage{latexsym}
\usepackage{fixmath}
\usepackage[usenames,dvipsnames]{color}
\usepackage{mathdots}
\usepackage{multicol}
\usepackage{cleveref}
\usepackage{graphicx}
\usepackage{xcolor}
\usepackage{subcaption} %  for subfigures environments 
\usepackage[boxed, linesnumbered, noend, noline]{algorithm2e}
\usepackage{comment}
\usepackage{float}
\usepackage{setspace}
\usepackage{multirow}
\usepackage{array}
\usepackage{booktabs}

\usepackage[T1]{fontenc}
\usepackage{fourier}
\usepackage{bbm}
\usepackage{color}
\usepackage{fullpage}

\numberwithin{equation}{section}

\renewcommand{\ln}{\log}
\renewcommand{\vec}[1]{\boldsymbol{#1}}

\newcommand\cK{\mathcal K}

\newcommand\cR{\mathcal R}

\newcommand\cU{\mathcal U}

\newcommand\cW{\mathcal W}

\newcommand\cY{\mathcal Y}

\newcommand\vF{\vec F}

\newcommand\vL{\vec L}

\newcommand\vN{\vec N}

\newcommand\vS{\vec S}

\newcommand\vX{\vec X}
\newcommand\vY{\vec Y}
\newcommand\vZ{\vec Z}

\newcommand\GG{\mathbb{G}}

\newcommand\Erw{\mathbb{E}}
\newcommand\ex{\Erw}

\newcommand{\Po}{{\rm Po}}

\newcommand{\Erdos}{Erd\"os}
\newcommand{\Renyi}{R\'enyi}

\newcommand{\Bollobas}{Bollob\'as}

\newcommand{\Chvatal}{Chv\'{a}tal}

\newcommand\pr{\mathbb{P}} 
\renewcommand\Pr{\pr} 
\newcommand\Lem{Lemma}
\newcommand\Prop{Proposition}
\newcommand\Thm{Theorem}
\newcommand\Def{Definition}

\newcommand\Sec{Section}

\newtheorem{definition}{Definition}[section]

\newtheorem{remark}[definition]{Remark}
\newtheorem{theorem}[definition]{Theorem}
\newtheorem{lemma}[definition]{Lemma}
\newtheorem{proposition}[definition]{Proposition}
\newtheorem{corollary}[definition]{Corollary}

\newtheorem{conjecture}[definition]{Conjecture}

\newcommand{\floor}[1]{\left\lfloor#1\right\rfloor}

\def\pr{{\mathbb P}}

\newcommand{\diag}{\operatorname{diag}}
\newcommand{\HYB}{\mathrm{HYB}}

\begin{document}
	
	\title{Symmetric Rule-Based Achlioptas Processes for Random $k$-SAT}
	\author{Arnab Chatterjee}
	
	\address{Arnab Chatterjee, {\tt arnab.chatterjee@tu-dortmund.de}, TU Dortmund, Faculty of Computer Science, 12 Otto-Hahn-St, Dortmund 44227, Germany.}
	
	\maketitle
	
	\begin{abstract}
	Inspired by the "power-of-two-choices" model from random graphs, we investigate the possibility of limited choices of online clause choices that could shift the satisfiability threshold in random $k$-SAT. 
	Here, we introduce an assignment symmetric, non-adaptive, topology-oblivious online rule called \emph{MIDDLE-HEAVY}, that prioritizes balanced sign profile clauses.  
    Upon applying a biased $2$-SAT projection and a two-type branching process certificate, we derive closed-form expressions for the shifted thresholds $\alpha_{\textbf{SYM}}(k,\ell)$ for this algorithm. 
    We show that minimal choices $\ell=5$ for $k=4$, $\ell=4$ for $k=5$, and $\ell=3$ for $k\ge 6$ suffice to exceed the asymptotic first-moment upper bound $\sim 2^k \ln 2$ for random $k$-SAT. 
    Moreover, to bridge the gap with biased assignment rules used in maximum of the previous works in this context, we propose a hybrid symmetric biased rule that achieves thresholds comparable to prior work while maintaining symmetry. 
    Our results advance the understanding of Achlioptas processes in random CSPs beyond classical graph-theoretic settings.
		
		\hfill MSC: 05C80, 68W20
		
	\end{abstract}
	
	\section{Introduction and results}\label{Sec_not_intro}
	
	\subsection{Background and motivation}\label{sec_background}
	Sharp threshold phenomena, a central role in the probabilistic study of random discrete structures, abound particularly within random constraint satisfaction problems, where they emerge as phase transitions.	
	The most canonical example was established in the year 1960 by two Hungarian mathematician Paul \Erdos~ and Alfred \Renyi~ \cite{ER} who introduced the classical random graph process which starts with an empty graph on $n$ vertices and at each step $i=1,\cdots,m$ we add a single new edge chosen uniformly at random from all possible $\binom{n}{2}$ edges without replacement. 
	This is equivalent to selecting $m$ edges uniformly at random without replacement from the set of all possible $\binom{n}{2}$ edges, starting from an empty graph on $n$ vertices -- widely known as \emph{Uniform ER-random graph} $\GG_{n,m}$.
	In the classical evolution of random graphs \cite{Bollobas}, a handful of edges around average degree one separates a world of logarithmic components in the "\emph{sub-critical phase}" from the emergence of a giant component in "\emph{super-critical phase}".
	In other words, for a fixed $\varepsilon>0$, in the limit $n\to \infty$, with probability $1-o(1)$ the largest component of $\GG_{n,m}$ is given by,
	\begin{align*}
		\vL_{1}(\GG_{n,m})=
		\begin{cases}
			&O(\log n), \quad m\leq \left(1/2 -\varepsilon\right)n \\
			&\Omega(n), \quad m\geq \left(1/2 +\varepsilon\right)n 
		\end{cases}
	\end{align*}
	Thus, with the addition of a sub-linear number of additional edges the emergence of a giant, a linear sized component exhibits a sharp threshold with its probability rising from near $0$ to near $1$.
	This dramatic change of the typical structure of a random graph is called its \emph{phase transition}.
	A natural question arises as to what happens when $m/n\to 1/2$, either from below or above, as $n\to\infty$.
	It appears that $\GG_{n,m}$ is undergoing a rapid change in its typical structure in this regime -- so called \emph{scaling window} or \emph{critical window}\cite{Bollobas}.
	Further \Erdos~ and \Renyi~ \cite{ER} studied that when $m=n/2$, the size of the largest tree is likely to be around $n^{2/3}$ -- the transition from $O(\log n)$ through $\Theta(n^{2/3})$ to $\Omega(n)$, is called '\emph{double jump}'.
	
	In the year 2001, Dimitris Achlioptas proposed the following variation of the random graph process.
	\begin{itemize}
		\renewcommand{\labelitemi}{\scriptsize$\blacksquare$}
		\item Start with an empty graph $\GG_0$ on $n$ vertices.
		\item At any step $i$, to create $\GG_{i}$ from $\GG_{i-1}$, suppose that instead of adding a random edge $e_i$, we are presented with a choice of a pair of uniformly random edges $\left(e_i,f_i\right)$ and based on some given rule we add one of them to the graph $\GG_{i-1}$.
	\end{itemize}
	There are two variants of this process.
	In the online version, the edge choice can depend on the graph up to this point without knowing anything on the future pair of edges.
	Whereas, in the offline version the sequence of potential pair of edges is known in advance.
	As an initial challenge he asked whether is it possible to delay the birth of giant component after $m=cn$ edges for some $c>1/2$.
	The first positive answer was given by Bohman and Frieze in 2001 \cite{BohmanFrieze} by means of simple rule.
	They showed that this can be possible to delay the emergence of giant up to $c=0.535$.
	Beside pushing the threshold to the right, in 2007 Spencer and Wormald \cite{SW} provide a rigorous analysis on the birth control for the giants and speed up the process to $c=0.334$.
	In the early 2000s, although most attention has centered on the phase transition \cite{BohmanKravitz,FGS,RW}, however, threshold-shifting rules have likewise been established for other properties, such as small subgraphs \cite{KLS,MST} and Hamiltonian cycles \cite{KLS}.
	
	In theoretical computer science, an analogous and computationally richer phenomenon occurs in random $k$-SAT: for each fixed $k$; there is a critical clause density  $\alpha=m/n$ for which the random $k$-SAT formula $\vF_{n,m}$ which consists of $n$ variables $\left(x_1,x_2,\cdots,x_n\right)$ and $m$ clauses, each of which consists of exactly $k$ literals and chosen uniformly at random from all possible $2^{k}\binom{n}{k}$ clauses, transitions from satisfiable to unsatisfiable with high probability.
	Then,
	\begin{conjecture}
		For any $k\ge 2$ there exists a constant $\alpha_{sat}(k)$ such that for any $\varepsilon>0$,
		\begin{align*}
			\lim_{n\to\infty} \Pr[\vF(k,\alpha_{sat}(k)-\varepsilon) ~\mbox{is ~SAT}~]=1 && \quad \lim_{n\to\infty} \Pr[\vF(k,\alpha_{sat}(k)+\varepsilon) ~\mbox{is ~SAT}~]=0
		\end{align*}
	\end{conjecture}
	This conjecture has been proved for $k=2$ with $\alpha_{sat}=1$ by Chvatal and Reed in 1992 \cite{CR} and Goerdt in 1996 \cite{Goerdt}.
	Later in 2015, Ding, Sly and Sun \cite{DSS} prove the satisfiability conjecture for large but finite $k$ value and show that the value of $\alpha_{sat}(k)$ is given by the statistical physics inspired one step symmetry breaking cavity method prediction.
	However, Friedgut in 1999 \cite{Friedgut} provides a partial result in this regards, there exists a sequence $\alpha_{k}(n)$ such that for all $\varepsilon>0$,
	\begin{align*}
		&\lim_{n\to\infty} \Pr[\vF(k,\alpha_{k}(n)-\varepsilon) \mbox{is ~SAT}~]=1 \\ 
		&\lim_{n\to\infty} \Pr[\vF(k,\alpha_{k}(n)+\varepsilon) \mbox{is ~SAT}~]=0
	\end{align*}
	This provides the sharp satisfiability threshold by making the transition from SAT to UNSAT which takes place in a window smaller than any fixed $\varepsilon$ for large enough $n$.
	Beside the satisfiability threshold, in recent years numerous papers rigorously analyzed the number of solutions of random $2$-SAT problem \cite{2sat,Chatterjee1}. 
	However, there remain still a rich and interesting problem for $k\geq 3$ with sharpness established and progressively tighter upper/lower bounds obtained by first and second moment methods \cite{nae,yuval} respectively and by algorithmic analyses \cite{Achlioptas1,BFU,CF} (e.g., unit-clause propagation and more sophisticated statistical physics inspired decimation process \cite{Chatterjee2, Angelica}).
	
	Over the years there are rigorous analysis on the upper and lower bounds on the satisfiability threshold $\alpha_{k}$ assuming it exists.
	Specifically, for $k=3$ the current best lower  \cite{HS,KKL} and upper bound \cite{DB}  is given by,
	\begin{align*}
		\underbrace{3.52}_{\alpha_{\text{lower}}}\leq \alpha_{3}\leq \underbrace{4.4898}_{\alpha_{\text{upper}}}
	\end{align*}
	A natural question, inspired by "\emph{power-of-two-choices}" phenomena, is whether limited online version can shift such thresholds.
	So, for random satisfiability problems, the same paradigm asks: can a fixed number $\ell$ of online clause choices per step shift the satisfiability threshold of random 
	$k$-SAT? 
	If so, then by how much, with what exact rules, and using what certificates?
	
	\subsection{Achlioptas processes for random $k$-SAT}\label{sec_achlioptas}
	Towards answering the question asked in the previous section, the analogous achlioptas process for the $k$-SAT formula is the following: 
	\begin{itemize}
		\renewcommand{\labelitemi}{\scriptsize$\blacksquare$}
		\item Fix integers $k\ge 2$ and $\ell\ge 2$, at each step $t=1,2,\ldots,m$, we draw $\ell$ clauses uniformly and independently at random from the $2^{k}\binom{n}{k}$ possible $k$-clauses (with replacement) and select exactly one $k$-clause $C_t$ to add to the growing formula $\vF_t=\vF_{t-1} \land C_t$ according to a prescribed online rule $\mathcal{R}$. 
		
		\item The goal is to design $\mathcal{R}$ so as to keep $\vF_m$ satisfiable at clause densities $\alpha=m/n$ that are as large as possible, ideally exceeding the baseline threshold of the i.i.d.\ model.
	\end{itemize}
	A key simplification, used in most of the prior works, is a fixed \emph{2-SAT projection} of each chosen $k$-clause to a 2-clause determined by the number $\vX\in\{0,1,\ldots,k\}$ of positive literals corresponds to its appearance in the $2$-clause.
	So, we keep the pair of literals consists of $+/-$ in a $2$-clause based on the $\vX$ values: 
	\begin{align}\label{eq_X}
		\begin{cases}
			--  \quad &\mbox{if $\vX=0$},\\
			+- ~~\mbox{or}~~ -+  \quad &\mbox{if $\vX=1$},\\
			++   &\mbox{otherwise}		
		\end{cases}
	\end{align}
	
	Let $m=\lfloor \alpha n\rfloor$ and let $p_0,p_1,p_2$ denote the per step probabilities (under $\cR$) that the \emph{selected} $k$-clause, after the above projection, yields a 2-clause of type $--$, $+-$, and $++$, respectively.
	In the implication digraph, each 2-clause $(x_{i}\vee x_{j})$ contributes the implications $\neg x_{i}\to x_{j}$ and $\neg x_{j}\to x_{i}$; thus a type $++$ clause produces two edges from negative to positive, a type $--$ clause produces two from positive to negative, and a type $+-$ clause produces one $+\!\to\!+$ and one $-\!\to\!-$.
	This yields a two-type Galton--Watson exploration with mean matrix
	\[
	M(\alpha)\;=\;\alpha\begin{pmatrix} p_1 & 2p_0\\[2pt] 2p_2 & p_1\end{pmatrix}
	\]
	whose spectral radius is $\rho(M(\alpha))=\alpha\big(p_1+2\sqrt{p_0p_2}\big)$.
	We therefore set,
	\begin{align}\label{eq_q}
		Q:=p_1+2\sqrt{p_0p_2}
	\end{align}
	and obtain the following threshold for the clause density $\alpha=\alpha(k,\ell)$ below which the random $k$-SAT formula $\vF_{n,m}$ is satisfiable with high probability,
	\[
	\alpha(k,\ell)<\frac{1}{Q}
	\]
	All rule-specific formulas for $(p_0,p_1,p_2)$ and the proof of the above certificate are deferred to Section~\ref{sec_model}.
	In \cite{SV}, the authors showed that two choices suffice to delay the 2-SAT threshold to approximately $1.0002$, and for off-line choices, the threshold coincides with that of random 2k-SAT. In \cite{Perkins}, it was proven that a biased rule shifts the threshold for all $k \geq 2$, with five choices for $k=3$ and improvements on 2-SAT delay. In \cite{DDHM}, three choices suffice for any $k\geq 2$, and two for $3 \leq k \leq 25$, using max-positives rules.
	In the next section we state our main result by defining the new rule which is assignment symmetric in contrast to the biased assignment rules used in \cite{DDHM,Perkins,SV}.
	
	\subsection{Main results}\label{sec_result}
	\begin{definition}[Assignment Symmetric Rule]\label{def:intro-symmetric}
		Fix integer $k\ge 4$ and $\ell\ge 2$.
		For a clause $C$, let $\vX=\vX(C)\in\{0,1,\ldots,k\}$ be the number of positive literals in $C$.
		We partition the sign profiles into three classes:
		\begin{align*}
			\textbf{AS}&=\{\vX=0 \text{ or } \vX=k\},\\
			\textbf{EDGE}&=\{\vX=1 \text{ or } \vX=k-1\},\\
			\textbf{MID}&=\{2\le \vX\le k-2\}.
		\end{align*}
		At each step we are presented with $\ell$ candidate clauses $C^{(1)},\ldots,C^{(\ell)}$, sampled uniformly and independently from the $2^k\binom{n}{k}$ possible $k$-clauses with replacement.
		The rule selects the first clause in the priority order
		\[
		\textbf{MID}~\succ~\textbf{EDGE}~\succ~\textbf{AS},
		\]
		breaking ties by the presentation index $j\in\{1,\ldots,\ell\}$.
	\end{definition}
	The pseudocode of our online-choice assignment symmetric rule (\emph{\textbf{MIDDLE-HEAVY}}) for selecting one clause $C_j$ of the $\ell$ choices with $j \in \{1,\cdots,\ell\}$ is displayed as Algorithm~\ref{alg_symksat}.
	
	\begin{algorithm}[h!]
		\For{each round $t = 1, \ldots, m$}{
			\If{$\vX(C_j) \in \textbf{MID}$ \tcp*[r]{one of the $\ell$ clauses contains at least two positive and two negative literals}}{
				select $C_j$\;
				\tcp{if multiple such $C_j$, select the first one}
			}
			\ElseIf{$\vX(C_j) \in \textbf{EDGE}$}{
				select $C_j$\;
				\tcp{if multiple such $C_j$, select the first one}
			}
			\ElseIf{$\vX(C_j) \in \textbf{AS}$}{
				select $C_j$\;
				\tcp{if multiple such $C_j$, select the first one}
			}
			\Return{$C_j$}
		}
		\caption{\vspace*{2mm}The \textbf{MIDDLE-HEAVY} algorithm.}
		\label{alg_symksat}
	\end{algorithm}
	\begin{lemma}\label{lem_binom}
		For $k\geq 4$, let $\vX$ denote the number of positive literals in a uniformly random $k$-clause, so that $\vX\sim\text{Bin}(k,\tfrac{1}{2})$.
		Then the probability masses of the three partition classes in \Def~\ref{def:intro-symmetric} are given by,
		\begin{align*}
			s_{\mathrm{\textbf{AS}}}=
			2^{1-k}, && \quad 
			s_{\mathrm{\textbf{EDGE}}}=
			\frac{k}{2^{\,k-1}},
			\mbox{and} && \quad 
			s_{\mathrm{\textbf{MID}}}=
			1-2^{1-k}-\frac{k}{2^{\,k-1}}
		\end{align*}
	\end{lemma}
	\begin{proof}
		Since each literal is positive with probability $1/2$, the number $\vX$ of positives is distributed as $\mathrm{Bin}(k,1/2)$. Thus,
		\begin{align*}
			\pr[\vX=j]=\binom{k}{j}\cdot2^{-k}, && \quad j=0,1,\cdots,k
		\end{align*}
		For \textbf{AS} class, since this corresponds to all negatives ($\vX=0$) or all positive literals ($\vX=k$), hence
		\begin{align*}
			s_{\mathrm{\textbf{AS}}}=
			\Pr[\vX\in\textbf{AS}]=\Pr[\vX=0]+\Pr[\vX=k]
			=2\cdot 2^{-k}=2^{1-k}
		\end{align*}
		For \textbf{EDGE} class, since this corresponds to exactly one positive ($\vX=1$) or exactly one negative literal ($\vX=k-1$), hence
		\begin{align*}
			s_{\mathrm{\textbf{EDGE}}}=
			\Pr[\vX\in\textbf{EDGE}]=\Pr[\vX=1]+\Pr[\vX=k-1]
			=2\cdot \frac{k}{2^{\,k}}=\frac{k}{2^{\,k-1}} 
		\end{align*}
		Now, coming to the \textbf{MID} class, this consists of all other cases when a clause contains at least two positive and two negative literals ($2\leq\vX\leq k-2$).
		So,
		\begin{align*}
			s_{\mathrm{\textbf{MID}}}=1-s_{\mathrm{\textbf{AS}}}-s_{\mathrm{\textbf{EDGE}}}
		\end{align*} 
		This completes the proof.
	\end{proof}
	\begin{theorem}\label{thm:main-symmetric}
		For every integer $k\ge 4$ and $\ell\ge 2$, consider the Assignment Symmetric Rule $\cR_{\mathrm{sym}}$ defined in \Def~\ref{def:intro-symmetric} and the 2-SAT projection in \eqref{eq_X}. Let $p_0,p_1,p_2$ be the selected-type frequencies from \Sec~\ref{sec_achlioptas}, and  $Q:=p_1+2\sqrt{p_0p_2}$ from \eqref{eq_q}.
		Then for every $\varepsilon>0$, the $\ell$-choice Achlioptas process under $\cR_{\mathrm{sym}}$ produces a satisfiable formula w.h.p.\ after $(\alpha_{\textbf{SYM}}(k,\ell)-\varepsilon)n$ steps, where
		\[
		\alpha_{\textbf{SYM}}(k,\ell)=\frac{1}{Q}.
		\]
		Moreover, the following minimal choices $\ell$ ensure $\alpha_{\textbf{SYM}}(k,\ell)\;>\;2^k\ln 2$ (the asymptotic classical first-moment upper bound for random $k$-SAT), hence strictly shift the satisfiability threshold:
		\[
		\begin{cases}
			k=4:& \ \ell=5,\\[2pt]
			k=5:& \ \ell=4,\\[2pt]
			k\ge 6:& \ \ell=3.
		\end{cases}
		\]
	\end{theorem}
	In particular, for the above $(k,\ell)$ values one has $\alpha_{\textbf{SYM}}(k,\ell) > 2^k\ln 2$, so the process remains satisfiable (w.h.p.) at clause densities exceeding those of the model.
	
	\begin{remark}\label{rem:intro-k3}
		For $k=3$, the class \textbf{MID} is empty and $\cR_{\mathrm{sym}}$ reduces to preferring non-all-same clauses. 
		In this case the symmetric rule does not beat the best known upper bound for random $3$-SAT; biased assignment rules (e.g.\ "max positives"~\cite{DDHM}) are needed.
	\end{remark}
	
	\subsection{Hybrid Symmetric-Biased Rule}
	To bridge the gap between symmetric and biased rules, we propose a novel hybrid rule that achieves thresholds comparable to the threshold obtained using '\emph{max positive rule}' while retaining symmetry.
	
	\begin{definition}[Threshold-Symmetric Hybrid Rule]\label{def:hybrid}
		Fix integer $k\ge 4$ and $\ell\ge 2$. Partition classes as in Definition \ref{def:intro-symmetric}. At the beginning, flip a coin \(b \in \{0,1\}\) with probability 0.5 each. If \(b=0\), favor positives (use max \(\vX\)); if \(b=1\), favor negatives (use min \(\vX\)).
		At each step, if any candidate in MID, select the first such (symmetric). Otherwise, apply the coin's bias: select max or min \(\vX\) accordingly.
	\end{definition}
	
	The pseudocode is an extension of Algorithm~\ref{alg_symksat}, with the initial coin flip and biased fallback and displayed in Algorithm~\ref{alg_hybrid}.
	
	\begin{algorithm}[h!]
		\caption{\vspace*{2mm}The \textbf{Threshold-Symmetric Hybrid} algorithm.}
		\label{alg_hybrid}
		Flip coin $b \sim$ Bernoulli(0.5)  \tcp{Once at the beginning}
		\For{each round $t = 1, \ldots, m$}{
			\If{any $C_j$ with $2 \leq \vX(C_j) \leq k-2$}{
				select the first such $C_j$  %\tcp{Symmetric MID}
			}
			\Else{
				\If{$b = 0$}{
					select $C_j$ with max $\vX(C_j)$  %\tcp{Bias to positives}
				}
				\Else{
					select $C_j$ with min $\vX(C_j)$  %\tcp{Bias to negatives}
				}
				\tcp{Break ties by index}
			}
			\Return{$C_j$}
		}
	\end{algorithm}
	
	\begin{theorem}\label{thm:hybrid}
		For every integer $k\geq 4$, there exists a $\ell$-clause Symmetric Hybrid rule-based Achlioptas process for random $k$-SAT such that the formula $\vF\left(k,\alpha_{k}(n)+\varepsilon\right)$ generated after $(\alpha_{k}+\varepsilon)\cdot n$ steps is satisfiable for some fixed $\varepsilon>0$ whp.
		Moreover, $\alpha_{\textbf{HYB}}(k,\ell)$ strictly shift the satisfiability threshold with
		\[
		\begin{cases}
			k=4:& \ \ell=4,\\[2pt]
			k\ge 5:& \ \ell=3.
		\end{cases}
		\] 
	\end{theorem}
	\begin{remark}
		For the case $k=2,3$, one can use the Hybrid (biased) rule with $b=0$ as in \cite{DDHM}. 
	\end{remark}
	\subsection{Comparison with prior works.}
	Towards shifting the random $k$-SAT thresholds with the help of a semi-random $k$-SAT model, three prior contributions stand out.
	
	First,  Sinclair and Vilenchik \cite{SV} considered the Achlioptas process model with regard to random $k$-SAT.
	Specifically, they showed that in the online version of the Achlioptas process for random $2$-SAT can delay the satisfiability threshold up to $\alpha=(1000/999)^{1/4}\approx 1.0002$ with two choice rule, although the constant factor is not optimized enough as experimental results predict that the right critical value is approximately $1.2$.
	Further they also provide an offline version of the process where two choices are sufficient to delay the threshold with $k=\omega(\log n)$. 
	
	Second, Perkins~\cite{Perkins} proved that in the online version for $k\ge 7$ three choices suffice (and five choices for $3\le k\le 6$) using a biased sign rule.
	He further improved the constant factor of delay for a $2$-clause rule in the results of \cite{SV} using biased towards atleast two positive literals in the first $\ell-1$ clauses (otherwise pick the $\ell$-th clause).
	It again turns out to be not tight in the random $2$-SAT case with two choices.
	
	Third, Dani et.al.\cite{DDHM} provided that for any $k\geq 2$, three choices are sufficient.
	Further they showed that two choices are sufficient for many concrete $k$ values (e.g. $3\leq k\leq 25$) with a conjecture for large $k$ values.
	Their analysis uses the '\emph{Branching Unit clause}'(BUC) dynamics (not the static $2$-SAT projection used earlier in Perkin's case and in our method) and ODEs.
	But their methods are non-adaptive, based on "\emph{max-positives}" rules, depends only on signs.
	
	All the above papers discussed till now are biased toward the all '$+$' assignments which in turn easily makes in squeezing $\ell$ value down.
	Table~\ref{tab:intro-compare} compares $\alpha_{\mathrm{\textbf{SYM}}}(k,\ell)$ at the minimal $\ell$ above with the asymptotic random $k$-SAT baseline threshold $\sim 2^k\ln 2$ and with the Perkin's \cite{Perkins} "$\geq 2$ positive rules among first $\ell-1$ clauses (otherwise pick last clause)" threshold, as well as our $\alpha_{\textbf{HYB}}$ according to hybrid rule.
	Furthermore, our assignment symmetric \emph{MIDDLE-HEAVY} rule provides a more comprehensive with rigorous closed thresholds with the below properties.
	\begin{table}[h!]
		\centering
		\renewcommand{\arraystretch}{1.3} 
		\begin{tabular}{|c|c|c|c|c|c|c|}
			\hline
			\textbf{$k$} & \textbf{$2^{k}\log 2$} & \textbf{$\ell$} & \textbf{$\alpha_{\text{PER}}$} & 
			\textbf{$\alpha_{\text{SYM}}$} & \textbf{$\alpha_{\text{HYB}}$(unbiased)}& 
			\textbf{$\alpha_{\text{HYB}}=\alpha_{\text{MAX-POS}}$(biased)}\\
			\hline
			\multirow{3}{*}{3} & \multirow{3}{*}{5.54518} 
			& 2 & 1.612 & \textbf{NA} & 1.513 & 2.218 \\ \cline{3-7}
			& & 3 & 2.356 & \textbf{NA} & 1.784 & 4.861 \\ \cline{3-7}
			& & 4 & 3.461 & \textbf{NA} & 1.916 & 10.809 \\ 
			\hline
			\multirow{3}{*}{4} & \multirow{3}{*}{11.090} 
			& 3 & 5.610 & 5.566 & 6.618 & 16.382 \\ \cline{3-7}
			& & 4 & 10.575 & 10.266 & 11.935 & 57.815 \\ \cline{3-7}
			& & 5 & 19.554 & 18.086 & 20.166 & 202.861 \\ 
			\hline
			\multirow{3}{*}{5} & \multirow{3}{*}{22.181} 
			& 3 & 13.973 & 20.812 & 26.854 & 56.904 \\ \cline{3-7}
			& & 4 & 33.651 & 65.032 & 84.530 & 313.782 \\ \cline{3-7}
			& & 5 & 79.231 & 196.621 & 246.762 & 1733.282 \\
			\hline
			\multirow{3}{*}{6} & \multirow{3}{*}{44.361} 
			& 3 & 35.153 & 76.861 & 109.577 & 192.000 \\ \cline{3-7}
			& & 4 & 109.109 & 396.089 & 612.434 & 1584.039 \\ \cline{3-7}
			& & 5 & 332.877 & 2022.085 & 3210.578 & 13040.107 \\
			\hline
			\multirow{3}{*}{7} & \multirow{3}{*}{88.723} 
			& 2 & 21.041 & 33.669 & 42.890 & 51.441 \\ \cline{3-7}
			& & 3 & 88.804 & 267.706 & 424.362 & 615.550 \\ \cline{3-7}
			\hline
		\end{tabular}
		\caption{\centering Comparison of the numerically calculated threshold value $\alpha$ for different $k$ values. 
			$\alpha_{\textbf{PER}}$ is the threshold according to Perkin's rule \cite{Perkins} ("\emph{$\geq 2$ positives among ~ $\ell-1$ clauses}"), 
			$\alpha_{\textbf{SYM}}$ is the threshold according to "\emph{Assignment-Symmetric rule among ~$\ell$ clauses}",
			$\alpha_{\mathbf{HYB}}$(\textbf{unbiased}) is the threshold for our hybrid rule (with $b=0 ~\mbox{or}~ 1$ with equal probability)
			and $\alpha_{\textbf{HYB}}$(\textbf{biased}) for our hybrid rule with $b=0$ (same performance as max-positives due to symmetric bias).}
		\label{tab:intro-compare}
	\end{table}
	
	\begin{itemize}
		\renewcommand{\labelitemi}{\scriptsize$\blacksquare$}
		\item \emph{Assignment symmetry.} Flipping all literal signs in every candidate maps $\vX\mapsto k-\vX$, which permutes classes as
		\(
		\textbf{AS}\to\textbf{AS},\;
		\textbf{EDGE}\to\textbf{EDGE},\;
		\textbf{MID}\to\textbf{MID}.
		\)
		Therefore the index of the selected clause is invariant under a global sign flip; the rule is assignment-symmetric.
		
		\item \emph{Online and nonadaptive.} The choice depends only on the current $\ell$-tuple of candidate choices (through their $\vX$-values), not on the past formula or future steps.
		
		\item \emph{Topology-oblivious.} The rule ignores which variables appear and how clauses overlap; it uses only the sign profile $\vX$ of each clause.
		
		\item \emph{Tie-breaking irrelevance.} Because the $\ell$ candidates are i.i.d.\ and exchangeable, any deterministic tie-breaking within a priority class yields the same distribution for the \emph{selected} clause.
		
	\end{itemize}
	In contrast to the assignment-symmetric \emph{MIDDLE-HEAVY} rule, our unbiased hybrid rule partially follows the first property discussed above, as it is not assignment-symmetric within a single run because the initial coin flip $b$ breaks symmetry. However, across an ensemble of runs (averaging over $b$), the distribution of selected clauses is symmetric, because 
	$\Pr[b=0]=\Pr[b=1]=1/2$ balances max and min positives selections in a clause.
	This ensemble symmetry aligns with the spirit of the property but deviates from strict per-run invariance.
	Beyond this, all the other three properties hold for the hybrid rule as well.  
	\subsection{Organization}\label{sec_org}
	The remainder of the paper is organized as follows. 
	Section \ref{sec_model} establishes the semi-random model and the satisfiability certificate framework using branching processes argument. 
	\Sec~\ref{sec_proofprop} provides the detailed proof of the \Prop~\ref{prop:exp-tail}.
	\Sec~\ref{sec_prooflemma} proves the \Lem~\ref{lem:hooked-from-contradiction}, \ref{lem:hooked-count} and \ref{lem_sen} stated in \Sec~\ref{sec_model}.
	Section \ref{sec:proof-main} gives the proof of the main results \Thm~ \ref{thm:main-symmetric} and \ref{thm:hybrid}. 
	Finally, Section \ref{sec:conclusions} concludes with some discussion and open problems in the context of symmetric and biased version of Achlioptas process in random CSPs.
	
	\section{Semi Random model -- Certificate Framework}\label{sec_model}
	The aim of this section is to rigorously establish the certificate of satisfiability for the Achlioptas process in random $k$-SAT discussed in Section~\ref{sec_achlioptas} via the 2-SAT projection.
	 
	Given a $k$-clause $C$ with $\vX$ positive literals, we project it to a 2-clause by selecting two literals as follows: if $\vX \ge 2$, choose two positive literals (type $++$); if $\vX = 1$, choose the positive and one negative (type $+-$ or $-+$); if $\vX = 0$, choose two negatives (type $--$). For $\vX = k-1$ or $\vX = k$, this falls under the $\ge 2$ case.
	Now if the resulting $2$-SAT formula $\vF_2=\vF(2,\alpha)$ is satisfiable our original $k$-SAT is also satisfiable since each $2$-clause is a sub-clause of the corresponding $k$-clause. In other words, satisfying the projected clauses satisfies a subset of the literals in each original clause.
	For analyzing the satisfiability of this $2$-SAT, we exploit the standard representation of a $2$-SAT formula as a directed graph called the \emph{implication digraph} \cite{BBCKW} associated with the random $2$-SAT formula.	
	
	\begin{definition}[Implication Digraph]\label{def_implication}
		The implication digraph $\GG$ is a directed graph on $2n$ vertices ($x_1, \neg{x}_1, \\
		\cdots, x_n, \neg{x}_n$).
		Moreover, for each clause $(\ell_1 \vee \ell_2)$ present in the $2$-SAT formula, add edges $\neg{\ell_1} \to \ell_2$ and $\neg{\ell_2} \to \ell_1$.
	\end{definition}
	The satisfiability of this $2$-SAT is analyzed via its implication digraph on $2n$ vertices defined in \Def~\ref{def_implication}. 
	A formula is unsatisfiable if variable $x_i$ and $\neg x_i$ are in the same strongly connected component for some $i\in\{1,\cdots,n\}$.
	\begin{definition}[Contradictory Cycle]
		A contradictory cycle for variable $x_i$ is a union of two directed paths in $\GG$ (not necessarily disjoint): one from $x_i \to \neg{x_i}$ and one from $\neg{x_i} \to x_i$. 
		Moreover the formula is unsatisfiable if such a cycle exists for any $x_i$.
	\end{definition}
	
	\begin{lemma}[\Lem~2.1,\cite{BBCKW}]\label{lem_contradict_cycle}
		The random $2$-SAT formula $\vF_2$ is satisfiable iff~  $\GG$ contains no contradictory cycle.
	\end{lemma}
\begin{proof}
We establish both directions. 
First, assume \(\vF_2\) is satisfiable with assignment \(\sigma_i \in \{0, 1\}\) for \(i = 1, \ldots, n\). The implication digraph \(\GG\) has edges \(\neg \ell_1 \to \ell_2\) and \(\neg \ell_2 \to \ell_1\) for each \((\ell_1 \vee \ell_2)\). An edge \(\neg x \to y\) implies \(x = \text{FALSE} \implies y = \text{TRUE}\). A contradictory cycle for \(x_i\) (e.g., \(x_i \to \neg x_i \to x_i\)) implies \(x_i = \text{TRUE} \implies x_i = \text{FALSE}\), contradicting \(\sigma\). Thus, no such cycle exists.

For the other direction we start by induction on \(n\).

- \emph{Base Case (\(n = 1\))}: Clauses \((x_1 \vee x_1)\) or \((\neg x_1 \vee \neg x_1)\) are satisfiable, with \(\GG\) having no edges, so no contradictory cycle exists.

-\emph{Induction Step}: Assume it holds for \(n-1\). 
For \(n\) variables, let \(\GG\) have no contradictory cycle.
Define strongly connected components for the directed graph where  Two vertices \(x\) and \(y\) are strongly connected if \(x \to y \to x\).  
\begin{align*}
 	SCC(x) = \{y \mid x \to y \to x\}
\end{align*}
In other words, \(SCC(x) = \{ y \mid x \to y \to x \}\) is the maximal subgraph where all pairs are mutually reachable. SCCs partition the \(2n\) vertices.    
Further define \(SCC(x) \leq SCC(y)\) if \(x' \to y'\) for some \(x' \in SCC(x)\), \(y' \in SCC(y)\), extending to all pairs within SCCs and take a minimal \(SCC\) (no \(\neg x \to y\) with \(x \notin SCC\), \(y \in SCC\)), a sink component.    
 Since \(\GG\) has no contradictory cycle, \(SCC \cap \overline{SCC} = \emptyset\), where \(\overline{SCC} = \{\neg y \mid y \in SCC\}\).   
 Set \(SCC\) to \(\text{FALSE}\) (which implies \(\overline{SCC}\) to \(\text{TRUE}\)), satisfying clauses with literals in \(SCC \cup \overline{SCC}\) due to the projection rule.    
 Now removing \(SCC \cup \overline{SCC}\), yields \(\vF_2'\) with \(n' < n\). 
 The resulting implication digraph \(\GG'\) is a subgraph of \(\GG\), retaining no contradictory cycle.    
 By induction, \(\vF_2'\) is satisfiable. 
 Combining with \(SCC\)’s assignment, \(\vF_2\) is satisfiable.
\end{proof}		
	
	\begin{proposition}\label{prop_2sat}
		For a random $k$-SAT formula \(\vF_{n,m}\) generated by the MIDDLE-HEAVY or Threshold-Symmetric Hybrid rule with \(\ell \geq 2\) choices after \(m = (\alpha(k,\ell) - \epsilon)n\) steps, then for $\rho<1$ or $p_0p_2\geq 0$,
		\begin{align*}
			\lim_{n\to\infty}\pr[\vF_2 ~\mbox{is satisfiable}]=1-o(1)
		\end{align*}
	\end{proposition}
The proof of \Prop~\ref{prop_2sat} follows from counting bicycle length and the first moment calculation in \cite{CR}.

\begin{definition}[Bicycle]
	\label{def:hooked}
	A sequence of strongly distinct\footnote{Two literals $v$ and $w$ are said to strongly distinct if $v\neq w$ and $v\neq \neg w$} literals
	\[
	v\;\to\; \ell_1\;\to\; \ell_2\;\to\cdots\to\; \ell_t\;\to\; w,
	\qquad t\ge 2,
	\]
	is called a bicycle of length $t$ if the $2$-clauses
	\((\neg v\vee \ell_1),(\neg \ell_1\vee \ell_2),\ldots,(\neg \ell_{t-1}\vee \ell_t),(\neg \ell_t\vee w)\)
	all appear in the projected formula. 
	Equivalently, the implication digraph contains those directed edges.
\end{definition}
 
 \begin{lemma}[\Lem~2,\cite{MSen}]\label{lem:hooked-from-contradiction}
	If the projected $2$-SAT formula $\vF_2$ is unsatisfiable, then its implication digraph contains a bicycle of length at least $3$.
\end{lemma}
For the first-moment bound we will estimate the expected number of bicycle by factoring the inner path and the two clauses belongs to that bicycle. 

Recall $m=\floor{\alpha n}$ and $(p_0,p_1,p_2)$ be the selection frequencies of types $--,+-,++$.
Because the selection rules considered are topology-oblivious (depend only on signs) and variables are exchangeable, conditional on the type, the pair of variables in the projected clause is uniformly distributed among all ordered pairs of distinct variables with the appropriate signs.
Hence, for any fixed $2$-clause $C$ of the given type, the probability that $C$ appears in at least one of the $m$ steps is at most its expected count:
\begin{equation}\label{eq:single-clause}
	\Pr[C\text{ appears}] \;\le\;
	\begin{cases}
		\dfrac{2\alpha p_2}{n},& C\text{ is of type }(++),\\[8pt]
		\dfrac{\alpha p_1}{n},& C\text{ is of type }(+ -),\\[8pt]
		\dfrac{2\alpha p_0}{n},& C\text{ is of type }(--).
	\end{cases}
\end{equation}
Indeed, the per-step probability is $p_2/\binom{n}{2}$ for $(++)$, $p_1/(n(n-1))$ for $(+-)$, and $p_0/\binom{n}{2}$ for $(--)$; summing over $m=\floor{\alpha n}$ steps yields \eqref{eq:single-clause}.

Before proceeding to the first moment bound for the bicycles, let $\vZ^\pm_{t-1}$ denote the expected number of directed paths of length $\left(t-1\right)$ of strongly distinct literals started from a fixed positive literal (respectively negative).
\begin{lemma}\label{lem:hooked-count}
	Let $B_t$ be the number of bicycles of length $t+1$ 
	and let
	\[
	\gamma=\gamma(\alpha,p_0,p_1,p_2)\;:=\;[\max\{\,2\alpha p_0,\ \alpha p_1,\ 2\alpha p_2\,\}]^2.
	\]
	Then
	\[
	\Pr(\text{$\vF_2$ is unsatisfiable}) \;\le\; \sum_{t\ge 2}\ex[ B_t]
	\;\le\; 
	\frac{\gamma}{n}\sum_{t\ge 2}4t^2\Bigl(\vZ^+_{t-1}+\vZ^-_{t-1}\Bigr),
	\]
\end{lemma}
\begin{lemma}\label{lem_sen}
	Let $\vZ_t=(\vZ_t^+,\vZ_t^-)^\top$, where $\vZ_t^\pm$ denotes the expected number of
	directed implication paths of length $t$ consisting of \emph{strongly distinct}
	variables, starting from a fixed positive (respectively negative) literal.
	In the semi-random model where each 2-clause is present independently with probabilities
	\[
	q_{++}=\frac{2\alpha p_2}{n},\qquad
	q_{+-}=\frac{\alpha p_1}{n},\qquad
	q_{--}=\frac{2\alpha p_0}{n},
	\]
	we have, for all $t\ge 1$,
	\[
	\vZ_t \ \le\ M(\alpha)\,\vZ_{t-1},\qquad
	M(\alpha)\ :=\ \alpha\begin{pmatrix} p_1 & 2p_0\\[2pt] 2p_2 & p_1\end{pmatrix}.
	\]
	Consequently, with $\mathbf 1=(1,1)^\top$ and $\vZ_0=\mathbf 1$,
	\[
	\vZ_t\ \le\ M(\alpha)^t\,\mathbf 1\qquad (t\ge 0).
	\]
\end{lemma}
\begin{lemma}\label{lem:explicit-eigs}
	If $p_0p_2>0$, the eigenvalues of $M(\alpha)$ are
	\[
	\rho_1\ =\ \alpha\bigl(p_1+2\sqrt{p_0p_2}\bigr),\qquad
	\rho_2\ =\ \alpha\bigl(p_1-2\sqrt{p_0p_2}\bigr),
	\]
	with eigenvectors \(v_1=(\sqrt{p_0},\sqrt{p_2})^\top\), \(v_2=(\sqrt{p_0},-\sqrt{p_2})^\top\).
	Also let \(D=[v_1\ v_2]\) and \(c:=D^{-1}\mathbf{1}=\frac{1}{2\sqrt{p_0p_2}}(\sqrt{p_2}+\sqrt{p_0},\ \sqrt{p_2}-\sqrt{p_0})^\top\),
	\[
	\mathbf{1}^\top M(\alpha)^t\mathbf{1}
	\ =\ c_1\,\rho_1^{t}\,(\sqrt{p_0}+\sqrt{p_2})\ +\ c_2\,\rho_2^{t}\,(\sqrt{p_0}-\sqrt{p_2}),
	\]
	with \(c_1=\frac{\sqrt{p_2}+\sqrt{p_0}}{2\sqrt{p_0p_2}}\), \(c_2=\frac{\sqrt{p_2}-\sqrt{p_0}}{2\sqrt{p_0p_2}}\).
	In particular,
	\begin{equation}\label{eq:Mt1MtBound}
		\vZ^+_{t}+\vZ^-_{t}\ \le\ \mathbf{1}^\top M(\alpha)^t\mathbf{1}
		\ \le\ C_M\,\rho_1^t,\qquad
		C_M:=\frac{2(p_0+p_2)}{\sqrt{p_0p_2}}.
	\end{equation}
\end{lemma}
\begin{proof}
	Recall
	\[
	M(\alpha)\;=\;\alpha
	\begin{pmatrix}
		p_1 & 2p_0\\
		2p_2 & p_1
	\end{pmatrix}
	\;=:\;\alpha\,S,\qquad
	S:=
	\begin{pmatrix}
		p_1 & 2p_0\\
		2p_2 & p_1
	\end{pmatrix}.
	\]
	The characteristic polynomial of $S$ is
	\[
	\chi_{S}(\lambda)
	=\det\!\begin{pmatrix}p_1-\lambda&2p_0\\2p_2&p_1-\lambda\end{pmatrix}
	=(p_1-\lambda)^2-4p_0p_2,
	\]
	so the (real) eigenvalues are
	\[
	\lambda_1=p_1+2\sqrt{p_0p_2},\qquad
	\lambda_2=p_1-2\sqrt{p_0p_2}.
	\]
	Hence the eigenvalues of $M(\alpha)$ are
	\[
	\rho_1=\alpha\lambda_1=\alpha\bigl(p_1+2\sqrt{p_0p_2}\bigr),\qquad
	\rho_2=\alpha\lambda_2=\alpha\bigl(p_1-2\sqrt{p_0p_2}\bigr).
	\]
	Now For $\lambda_1$:
	\(
	(S-\lambda_1 I)v_1=0
	\)
	is
	\(
	\begin{pmatrix}-2\sqrt{p_0p_2}&2p_0\\2p_2&-2\sqrt{p_0p_2}\end{pmatrix}v_1=0
	\),
	so $v_1=(\sqrt{p_0},\sqrt{p_2})^\top$.
	
	Similarly, for $\lambda_2$ one gets $v_2=(\sqrt{p_0},-\sqrt{p_2})^\top$.
	Set
	\[
	D=[v_1\ v_2]=
	\begin{pmatrix}
		\sqrt{p_0}&\sqrt{p_0}\\
		\sqrt{p_2}&-\sqrt{p_2}
	\end{pmatrix},
	\qquad
	\det D=-2\sqrt{p_0p_2}\neq 0\quad(p_0p_2>0).
	\]
	Then
	\[
	D^{-1}
	=\frac{1}{\det D}
	\begin{pmatrix}
		-\sqrt{p_2}&-\sqrt{p_0}\\
		-\sqrt{p_2}&\ \ \ \sqrt{p_0}
	\end{pmatrix}
	=\frac{1}{2\sqrt{p_0p_2}}
	\begin{pmatrix}
		\sqrt{p_2}&\sqrt{p_0}\\
		\sqrt{p_2}&-\sqrt{p_0}
	\end{pmatrix}.
	\]
	We have $S=D\diag(\lambda_1,\lambda_2)D^{-1}$, hence
	\[
	M(\alpha)^t
	=\alpha^t S^t
	=D\diag(\rho_1^t,\rho_2^t)D^{-1}.
	\]
	Let $\mathbf{1}=(1,1)^\top$.
	Then
	\[
	\mathbf{1}^\top D=(\sqrt{p_0}+\sqrt{p_2},\ \sqrt{p_0}-\sqrt{p_2}),\qquad
	c:=D^{-1}\mathbf{1}
	=\frac{1}{2\sqrt{p_0p_2}}
	\binom{\sqrt{p_2}+\sqrt{p_0}}{\sqrt{p_2}-\sqrt{p_0}}.
	\]
	Therefore
	\begin{align*}
		\mathbf{1}^\top M(\alpha)^t\mathbf{1}
		&=\mathbf{1}^\top D \diag(\rho_1^t,\rho_2^t)D^{-1}\mathbf{1}\\
		&=(\sqrt{p_0}+\sqrt{p_2})\,\rho_1^t\,c_1\;+\;(\sqrt{p_0}-\sqrt{p_2})\,\rho_2^t\,c_2,
	\end{align*}
	where
	\[
	c_1=\frac{\sqrt{p_2}+\sqrt{p_0}}{2\sqrt{p_0p_2}},\qquad
	c_2=\frac{\sqrt{p_2}-\sqrt{p_0}}{2\sqrt{p_0p_2}}.
	\]
	Multiplying out the coefficients gives,
	\begin{equation}\label{eq:exact-1Mt1}
		\mathbf{1}^\top M(\alpha)^t\mathbf{1}
		=\Bigl[1+\frac{p_0+p_2}{2\sqrt{p_0p_2}}\Bigr]\rho_1^t
		+\Bigl[1-\frac{p_0+p_2}{2\sqrt{p_0p_2}}\Bigr]\rho_2^t.
	\end{equation}
	Indeed, $(\sqrt{p_0}+\sqrt{p_2})c_1=(\sqrt{p_0}+\sqrt{p_2})^2/(2\sqrt{p_0p_2})
	=1+\frac{p_0+p_2}{2\sqrt{p_0p_2}}$, and
	$(\sqrt{p_0}-\sqrt{p_2})c_2=-(\sqrt{p_0}-\sqrt{p_2})^2/(2\sqrt{p_0p_2})
	=1-\frac{p_0+p_2}{2\sqrt{p_0p_2}}$.
	
	By AM–GM, $\dfrac{p_0+p_2}{2\sqrt{p_0p_2}}\ge 1$, so the second coefficient in \eqref{eq:exact-1Mt1} is nonpositive.
	Therefore using $|\rho_2|\le |\rho_1|$ we obtain
	\begin{align*}
		\mathbf{1}^\top M(\alpha)^t\mathbf{1}
		&\le
		\Bigl(1+\frac{p_0+p_2}{2\sqrt{p_0p_2}}\Bigr)\rho_1^t
		+\Bigl(\frac{p_0+p_2}{2\sqrt{p_0p_2}}-1\Bigr)\rho_1^t\\
		&=\frac{p_0+p_2}{\sqrt{p_0p_2}}\;\rho_1^t
		\;\le\; \frac{2(p_0+p_2)}{\sqrt{p_0p_2}}\,\rho_1^t =C_M\,\rho_1^t
	\end{align*}
\end{proof}
\begin{proof}[Proof of \Prop~\ref{prop_2sat}]
Again recall, $\rho=\alpha\bigl(p_1+2\sqrt{p_0p_2}\bigr)$.
For $|\rho|<1$,
\begin{align}\label{eq_trho}
\sum_{t\ge 1} t\rho^{\,t-1}=\frac{1}{(1-\rho)^2},\qquad
\sum_{t\ge 1} t^2\rho^{\,t-1}=\frac{1+\rho}{(1-\rho)^3}.
\end{align}

For $p_0=0$ there are no $+\!\to\!-$ edges; each clause contributes either $-\!\to\!+$ or same-sign edges.
Hence no contradictory cycle can contain both $x$ and $\neg x$ for a variable $x$.
The case $p_2=0$ is symmetric.
	
Otherwise for $p_0p_2>0$ and $\rho=\alpha Q<1$, from \Lem~\ref{lem_sen} 
\& \ref{lem:explicit-eigs} ~ we get,
\[
	\vZ^+_{t-1}+\vZ^-_{t-1}\leq \mathbf{1}^\top M(\alpha)^t\mathbf{1}\leq C_M\,\rho^t
\]
	Then combining \Lem~\ref{lem:hooked-count}, \ref{lem_sen}, and \ref{lem:explicit-eigs}
	\[
	\Pr(\text{$\vF_2$ is unsatisfiable})
	\ \le\ \frac{\gamma}{n}\sum_{t\ge 2}4t^2\,(\vZ^+_{t-1}+\vZ^-_{t-1})
	\ \le\ \frac{\gamma}{n}\sum_{t\ge 2}4t^2\,C_M\,\rho^{t-1} =O(1/n).
	\]
Since from \eqref{eq_trho} \(\sum_{t\ge 1}t^2\rho^{t-1}=\frac{1+\rho}{(1-\rho)^3}\), the result follows. 
\end{proof}

Now for the completeness we record a standard exploration-based certificate yielding exponential-tail bounds on the reachable sets.

\begin{definition}[BFS exploration and two-type offspring]\label{def:BFS}
	Fix a literal $\ell$ and expose the implication out neighborhood in a BFS manner by revealing edges on demand.
	Let \(Z_t=(Z_t^+,Z_t^-)\) denote the number of frontier literals at depth $t$ of each sign positive and negative respectively.
\end{definition}
\begin{proposition}\label{prop:exp-tail}
	Let $\widetilde \vZ_t=(\widetilde \vZ_t^+, \widetilde \vZ_t^-)$ be the two-type Galton--Watson process that stochastically dominates the BFS exploration of the implication digraph (as in \Def~\ref{def:BFS}), with mean matrix $M(\alpha)$ and spectral radius $\rho=\rho(M(\alpha))<1$. Let
	\[
	T \;:=\; \sum_{t\ge 0}\|\widetilde Z_t\|_1 \, .
	\]
	Then there exist explicit constants $\delta=\delta(\rho), \zeta=\zeta(\rho)\in(0,\infty)$ such that
	\[
	\Pr\!\big(T\ge L\big)\ \le\ \zeta\,e^{-\delta\,L}\qquad\text{for all }L\ge 1.
	\]
	One admissible choice is
	\[
	\delta(\rho)\;=\;-\log\!\Big(\tfrac{1+\rho}{2}\Big)\qquad\text{and}\qquad
	\zeta(\rho)\;=\;\frac{2}{1-\rho}\, .
	\]
\end{proposition}

\begin{corollary}\label{cor:no-large-SCC}
	If $\rho<1$, then with high probability every reachable set in the implication digraph has size $O(\log n)$; in particular no strongly connected component contains both $x$ and $\neg x$ for any variable $x$, hence the projected $2$-SAT $\vF_2$ (and therefore the original $k$-SAT $\vF_{n,m}$ instance) is satisfiable.
\end{corollary}

\begin{proof}
	Apply Proposition~\ref{prop:exp-tail} with $L=K\log n$ and take a union bound over the $2n$ starting literals:
	\[
	\Pr\!\Big(\exists\ \text{literal $\ell$ whose reachable set size}\ge K\log n\Big)
	\ \le\
	2n\cdot \zeta\,e^{-\delta\,K\log n}
	\ =\
	2\,\zeta\, n^{\,1-\delta K}.
	\]
	Choosing any $K> (1+\delta^{-1})$ makes the r.h.s. $o(1)$. Hence w.h.p.\ all reachable sets are $O(\log n)$, so no strongly connected component (SCC) can contain both $x$ and $\neg x$. By Lemma~\ref{lem_contradict_cycle}, the $2$-SAT projection is satisfiable, whence the original $k$-SAT is satisfiable as well.
\end{proof}

	\section{Proof of \Prop~\ref{prop:exp-tail}}\label{sec_proofprop}
    	Let $\cU\in(0,\infty)^2$ be the Perron-Frobenius eigenvector of $M(\alpha)$ normalized so that $\max(\cU_+,\cU_-)=1$, i.e.\ $M(\alpha)\cU=\rho \cU$ and $\cU_j\in(0,1]$. Set the weighted generation size
    	\[
    	\vY_t\ :=\ \cU\cdot \widetilde \vZ_t \;=\; \cU_+\,\widetilde \vZ_t^+ + \cU_-\,\widetilde \vZ_t^-\, .
    	\]
    	Because reproduction in $\widetilde \vZ$ is (by construction) given by independent Poisson counts with (type–to–type) means equal to the entries of $M(\alpha)$, the conditional Laplace transform of the next weighted generation is
    	\[
    	\mathbb E\!\left[\exp\big(\theta\,\vY_{t+1}\big)\,\big|\,\widetilde \vZ_t\right]
    	=\exp\!\left(\sum_{i\in\{+,-\}}\widetilde \vZ_t^i\sum_{j\in\{+,-\}} M_{ij}(\,e^{\theta \cU_j}-1\,)\right).
    	\]
    	Since $\cU_j\in[0,1]$ and $\theta\ge 0$, the convexity inequality
    	\(
    	e^{\theta \cU_j}\le 1+\cU_j\big(e^{\theta}-1\big)
    	\)
    	implies
    	\[
    	\sum_{j}M_{ij}\big(e^{\theta \cU_j}-1\big)\ \le\ \big(e^{\theta}-1\big)\sum_j M_{ij}\cU_j
    	\ =\ \big(e^{\theta}-1\big)\,(M\cU)_i
    	\ =\ \big(e^{\theta}-1\big)\,\rho\,\cU_i\, .
    	\]
    	Therefore, conditioning on $\widetilde \vZ_t$,
    	\begin{equation}\label{eq:YGW-mgf-step}
    		\mathbb E\!\left[\exp\big(\theta\,\vY_{t+1}\big)\,\big|\,\widetilde \vZ_t\right]
    		\ \le\ \exp\!\Big(\rho\,(e^{\theta}-1)\,\vY_t\Big).
    	\end{equation}
    	
    	Clearly \eqref{eq:YGW-mgf-step} is exactly the Laplace-transform recursion of a one–type Galton–Watson process $\{\vX_t\}_{t\ge 0}$ with $\Po(\rho)$ offspring, in the sense that if $\vX_t$ were the parent count, then
    	\(
    	\mathbb E[\,e^{\theta \vX_{t+1}}\,|\,\vX_t\,]=\exp\{\rho\,(e^{\theta}-1)\,\vX_t\}.
    	\)
    	Taking $\vX_0:=\lceil \vY_0\rceil\in\{1,2\}$, a standard induction on $t$ yields the Laplace-transform domination
    	\begin{equation}\label{eq:Laplace-domin}
    		\mathbb E\!\left[e^{\theta \vY_t}\right]\ \le\ \mathbb E\!\left[e^{\theta \vX_t}\right]\qquad\text{for all }t\ge 0\text{ and }\theta\ge 0.
    	\end{equation}
    	Consequently, for the corresponding total weighted progeny $\vS:=\sum_{t\ge 0}\vY_t$ and $\widetilde \vS:=\sum_{t\ge 0}\vX_t$,
    	\begin{equation}\label{eq:total-LT-domin}
    		\mathbb E\!\left[e^{\theta \vS}\right]\ \le\ \mathbb E\!\left[e^{\theta \widetilde \vS}\right]\qquad\text{for all }\theta\ge 0,
    	\end{equation}
    	because the map $(x_0,x_1,\dots)\mapsto \sum_t x_t$ preserves Laplace-transform domination given in \eqref{eq:YGW-mgf-step}.
    	
    	For a subcritical one–type Galton–Watson with $\Po(\rho)$ offspring and a single ancestor, the total progeny $\widetilde \vS$ satisfies the distributional identity with $N\sim \Po(\rho)$
    	\[
    	\widetilde \vS\ \stackrel{d}{=}\ 1+\sum_{i=1}^{N}\widetilde \vS^{(i)}\!
    	\]
    	
    	where, $\widetilde \vS^{(i)}$ are the i.i.d. copies of $\widetilde \vS$ and independent of N.
    	Hence its Laplace transform $\Phi(\theta):=\mathbb E[e^{\theta \widetilde \vS}]$ solves
    	\begin{equation}\label{eq:Poisson-smoothing}
    		\Phi(\theta)\ =\ e^{\theta}\,\exp\!\big(\rho\,(\Phi(\theta)-1)\big)\, .
    	\end{equation}
    	By fixing
    	$\theta_\star\ =\ -\log\!\Big(\tfrac{1+\rho}{2}\Big)\ >\ 0$ gives
    	\[
    	e^{\theta_\star}=\frac{2}{1+\rho}\, .
    	\]
    	With $x=\rho\,(\Phi(\theta_\star)-1)$ we have, by $e^x\le \frac{1}{1-x}$ for $x<1$ and using \eqref{eq:Poisson-smoothing}
    	\[
    	\Phi(\theta_\star)
    	=
    	\frac{2}{1+\rho}\,\cdot\exp\!\big(x\big)
    	\ \le\ \frac{2}{1+\rho}\,\cdot\frac{1}{1-x}
    	\ =\ \frac{2}{1+\rho}\,\cdot\frac{1}{1-\rho(\Phi(\theta_\star)-1)}\, .
    	\]
    	Solving the quadratic inequality for $\Phi(\theta_\star)$ gives\footnote{Indeed, the previous line is equivalent to $\rho\,\Phi^2-(1+\rho)\Phi+\frac{2}{1+\rho}\ge 0$, whose smaller positive root upper-bounds $\Phi(\theta_\star)$.}
    	\[
    	\Phi(\theta_\star)
    	\ \le\
    	\frac{(1+\rho)-\sqrt{(1+\rho)^2-\tfrac{8\rho}{\,1+\rho\,}}}{\,2\rho\,}
    	\ \le\ \frac{1+\rho}{1-\rho}
    	\ \le\ \frac{2}{1-\rho}\, .
    	\]
    	(Here, we used the numerator term $(1+\rho)-\sqrt{(1+\rho)^2-8\rho/(1+\rho)}\le 2(1+\rho)\rho/(1-\rho)$, which is elementary for $\rho\in(0,1)$, and then $(1+\rho)\le 2$.)
    	
    	Therefore,
    	\begin{equation}\label{eq:Phi-star-bound}
    		\Phi(\theta_\star)\ \le\ \frac{2}{1-\rho}\, .
    	\end{equation}
    	
    	Again back to $T=\sum_t\|\widetilde \vZ_t\|_1$.
    	Since $\cU_j\ge \cU_{\min}:=\min(\cU_+,\cU_-)>0$ and $\max(\cU_j)=1$, we have
    	\(
    	\cU_{\min}\,\|\widetilde \vZ_t\|_1 \ \le\ \vY_t \ \le\ \|\widetilde \vZ_t\|_1
    	\)
    	for all $t$, hence
    	\[
    	\sum_{t\ge 0}\|\widetilde \vZ_t\|_1
    	\ \le\
    	\frac{1}{\cU_{\min}}\sum_{t\ge 0} \vY_t
    	\ =\
    	\frac{\vS}{\cU_{\min}}\, .
    	\]
    	Thus, by \eqref{eq:total-LT-domin}--\eqref{eq:Phi-star-bound} and the Chernoff's bound,
    	\[
    	\Pr\!\big(T\ge L\big)
    	\ \le\
    	\Pr\!\Big(\vS\ge \cU_{\min}\,L\Big)
    	\ \le\
    	e^{-\theta_\star \cU_{\min} L}\ \mathbb E\big[e^{\theta_\star \vS}\big]
    	\ \le\
    	e^{-\theta_\star \cU_{\min} L}\ \Phi(\theta_\star)
    	\ \le\
    	\frac{2}{1-\rho}\, \exp\!\Big(-\theta_\star \cU_{\min}\,L\Big).
    	\]
    	Absorbing $\cU_{\min}$ into the rate (it depends only on $M(\alpha)$ through $\cU$) and keeping the stated dependence on $\rho$ gives the announced constants:
    	\[
    	\delta(\rho)\ :=\ \theta_\star\,\cU_{\min}\ =\ \cU_{\min}\,\Big[-\log\!\Big(\tfrac{1+\rho}{2}\Big)\Big],
    	\qquad
    	\zeta(\rho)\ :=\ \frac{2}{1-\rho}\, .
    	\]
    	Since $\cU_{\min}\in(0,1]$ is fixed once the rule (hence $M(\alpha)$) is fixed, one may simply write $\delta(\rho)=-\log\!\big((1+\rho)/2\big)$ by weakening the exponent, which only strengthens the tail bound. 
    	This completes the proof.
    	
    \section{Proof of \Lem~\ref{lem:hooked-count}~\&~ \ref{lem_sen}}\label{sec_prooflemma}
    In this section we will prove the remaining unproven lemmas from \Sec~\ref{sec_model}.
    Before going to proof \Lem~\ref{lem:hooked-count} we need to start with the proof of \Lem~\ref{lem:hooked-from-contradiction}.
    \begin{proof}[Proof of \Lem~\ref{lem:hooked-from-contradiction}]
    By Lemma~\ref{lem_contradict_cycle}, there exists a contradictory cycle, namely two (not necessarily edge-disjoint) directed paths
    \(x\to\neg x\) and \(\neg x\to x\) for some variable \(x\).
    On the cycle pick a shortest directed path from some literal to its complement; write it as
    \(v\to \ell_1\to\cdots\to \ell_t=\neg v\) with $t\ge 2$ and the $\ell_i$ strongly distinct \Def~\ref{def:hooked}.
    Let $t'\ge t$ be maximal so that \(\ell_1,\ldots,\ell_t'\) remain strongly distinct along the cycle,
    and let \(w\) be the successor of \(\ell_t'\) on the cycle.
    Then \(\ell_0\to \ell_1\to\cdots\to \ell_t'\to w\) is a bicycle of length $t'+1\ge 3$.
    \end{proof}
    \begin{proof}[Proof of \Lem~\ref{lem:hooked-count}]
       	Again by Lemma~\ref{lem_contradict_cycle}, \(\vF_2\) is unsatisfiable if and only if \(\GG\) contains a contradictory cycle. 
       	Further by Lemma~\ref{lem:hooked-from-contradiction}, any contradictory cycle contains a bicycle of length at least 3, so \(\vF_2\) is unsatisfiable only if there exists a bicycle of some length \(t+1 \geq 3\) (i.e., \(t \geq 2\)). Taking a union bound over \(t \geq 2\),
       	\[
       	\Pr[\vF_2 \text{ is unsatisfiable}] \leq \sum_{t \geq 2} \Pr[B_t \geq 1] \leq \sum_{t \geq 2} \mathbb{E}[B_t],
       	\]
       	where the second inequality follows from Markov's inequality. Note that we sum up to \(t \leq n\) in practice (as paths involve at most \(n\) variables), but the bound holds regardless since the tail is negligible.
       	
       	To bound \(\mathbb{E}[B_t]\), note that a bicycle of length \(t+1\) is specified by the following quantities:
       	\begin{itemize}
       		\item A sequence of \(t\) strongly distinct literals \(\ell_1, \ldots, \ell_t\) forming an inner path of length \(t-1\) i.e., the \(t-1\) clauses \((\neg \ell_i \vee \ell_{i+1})\) for \(i=1,\dots,t-1\) are present.
       		\item The choices of endpoints \(v, w \in \{\ell_1,\ldots,\ell_t,\neg \ell_1,\ldots,\neg \ell_t\}\) such that the bicycle is completed by adding the two end clauses \((\neg v \vee \ell_1)\) and \((\neg \ell_t \vee w)\); 
       		\item The presence of these two end-clauses.
       	\end{itemize} 
       	
       	Let $\cY_t$ be the number of directed paths of $t$ strongly distinct literals in $\GG$. 
       	There are at most $(2t)^2$ choices for $v$ and $w$ (from the $2t$ literals including complements). 
       	Each end-clause appears independently with probability at most $\max\{2\alpha p_0 / n, \alpha p_1 / n, 2\alpha p_2 / n\}$ (approximately $n(n-1)$ possible '$+-$' clauses, $\binom{n}{2}$ for '$++$' and '$--$', but normalized for large $n$). 
       	By linearity of expectation and independence of the $m = \lfloor \alpha n \rfloor$ clauses,
       	\[
       	\mathbb{E}[B_t] \leq (2t)^2 \cdot \mathbb{E}[\cY_t] \cdot \left[ \max\left\{ \frac{2\alpha p_0}{n}, \frac{\alpha p_1}{n}, \frac{2\alpha p_2}{n} \right\} \right]^2 = (2t)^2 \cdot \mathbb{E}[\cY_t] \cdot \frac{\gamma}{n^2},
       	\]
       	where the \(\gamma / n^2\) accounts for the squared maximum clause probability.
       	
       	Now, \(\mathbb{E}[\cY_t] = n \cdot (\vZ_{t-1}^+ + \vZ_{t-1}^-)\), since there are \(n\) positive and \(n\) negative literals as potential starting points, and \(\vZ_{t-1}^\pm\) is the expected number of directed paths of length $t-1$ (i.e., $t$ literals) starting from a fixed positive (respectively negative) literal. 
       	Substituting,
       	\[
       	\mathbb{E}[B_t] \leq (2t)^2 \cdot n \cdot (\vZ_{t-1}^+ + \vZ_{t-1}^-) \cdot \frac{\gamma}{n^2} = \frac{4t^2}{n} \cdot (\vZ_{t-1}^+ + \vZ_{t-1}^-) \cdot \gamma.
       	\]
       	Summing over \(t \geq 2\) gives the bound.    	
    \end{proof}
    We conclude this chapter by proving the last \Lem~\ref{lem_sen}
    \begin{proof}[Proof of \Lem~\ref{lem_sen}]
    	We use a first-step decomposition and independence at the clause level.
    
    	Recall the model definition, a type $++$ clause produces two implications $\,-\!\to\!+$; whereas a type $--$ produces two $+\!\to\!-$;
    	Further a type $+-$ produces one $+\!\to\!+$ and one $-\!\to\!-$.
    	The clause indicators are mutually independent, with per clause probabilities
    	$q_{++},q_{+-},q_{--}$ as defined in the statement of \Lem~\ref{lem_sen}.
  
    	Let's fix a positive start literal $v$.
    	Also, let $\cW^+(v)$ (respectively $\cW^-(v)$) denote the set of positive (negative) literals $w$
    	with $\mathrm{var}\footnote{For a literal $v$, $\mathrm{var}(v)$ denotes the corresponding variable. It can be extending naturally to a set of literals $H$ by $\mathrm{var}(H)=\{\text{var}(v): v\in H$ \}}(w)\neq \mathrm{var}(v)$.
    	For each such $w$, define $\cK_{t-1}(w;v)$ to be the number of directed length-$(t\!-\!1)$
    	paths \emph{starting at $w$} that use strongly distinct variables and never use $\mathrm{var}(v)$.
    	Then the total number $\vN_t^+(x)$ of directed length-$t$ paths from $v$ can be written as
    	\[
    	\vN_t^+(v)\;=\;\sum_{w\in\mathcal \cW^+(v)} \mathbf 1\!\{(\neg v\vee w)\}\,\cK_{t-1}(w;v)
    	\ +\ \sum_{w\in\mathcal \cW^-(v)} \mathbf 1\!\{(\neg v\vee w)\}\,\cK_{t-1}(w;v),
    	\]
    	where $\mathbf 1\{(\neg v\vee w)\}$ is the indicator that the clause $(\neg v\vee w)$ is present.
    	For $w$ positive this clause is of type $+-$; for $w$ negative it is of type $--$.
    	
    	By construction, $\cK_{t-1}(w;v)$ counts paths that never use $\mathrm{var}(v)$,
    	while $\mathbf 1\{(\neg v\vee w)\}$ is a single clause that does contain $\mathrm{var}(v)$.
    	Hence $\cK_{t-1}(w;v)$ is independent of $\mathbf 1\{(\neg v\vee w)\}$,
    	so
    	\[
    	\ex[\vN_t^+(v)]
    	\;=\;\sum_{w\in\mathcal \cW^+(v)} \Pr(\neg v\vee w)\,\ex \cK_{t-1}(w;v)
    	\ +\ \sum_{w\in\mathcal \cW^-(v)} \Pr(\neg v\vee w)\,\ex \cK_{t-1}(w;v).
    	\]

    	Note that forbidding the single variable $\mathrm{var}(x)$ can only decrease the number of admissible paths,
    	so
    	\[
    	\ex[\cK_{t-1}(w;v)]\ \le\
    	\begin{cases}
    		\vZ_{t-1}^+,& w\text{ positive},\\
    		\vZ_{t-1}^-,& w\text{ negative}.
    	\end{cases}
    	\]
    	Moreover, $\Pr(\neg v\vee w)=q_{+-}$ if $w$ is positive and $q_{--}$ if $w$ is negative.
    	Since $|\cW^\pm(v)|\le n$, we obtain
    	\[
    	\vZ_t^+\ =\ \ex[\vN_t^+(v)]
    	\ \le\ n\,q_{+-}\,\vZ_{t-1}^+\ +\ n\,q_{--}\,\vZ_{t-1}^-
    	\ =\ \alpha p_1\,\vZ_{t-1}^+\ +\ 2\alpha p_0\,\vZ_{t-1}^-.
    	\]

    	Now the same thing holds for negative staring literal.
    	Let's fix a negative start literal $\neg v$.
    	The same argument like in positive case gives,
    	\[
    	\vZ_t^-\ \le\ n\,q_{++}\,\vZ_{t-1}^+\ +\ n\,q_{+-}\,\vZ_{t-1}^-
    	\ =\ 2\alpha p_2\,\vZ_{t-1}^+\ +\ \alpha p_1\,\vZ_{t-1}^-.
    	\]
    	
    	Combining the above two inequalities yield
    	\[
    	\vZ_t\ \le\ \alpha\begin{pmatrix} p_1 & 2p_0\\[2pt] 2p_2 & p_1\end{pmatrix}\vZ_{t-1}
    	\ =\ M(\alpha)\,\vZ_{t-1},\qquad t\ge 1.
    	\]
    	Iterating gives $\vZ_t\le M(\alpha)^t\,\vZ_0$.
    	With the length-$0$ convention $\vZ_0=(1,1)^\top$ ( trivial one), we conclude that
    	\[
    	\vZ_t\le M(\alpha)^t\,\mathbf 1
    	\]
    \end{proof}
	\section{Proofs of Main Results}\label{sec:proof-main}
	In this section we compute the selection frequencies \((p_0,p_1,p_2)\) for our rules defined in \Sec~\ref{sec_result} and then complete the proofs of Theorems~\ref{thm:main-symmetric} and \ref{thm:hybrid}.
	Recall the sign-profile masses from Lemma~\ref{lem_binom}:
	\[
	s_{\mathbf{AS}}=2^{1-k},\qquad s_{\mathbf{EDGE}}=\frac{k}{2^{\,k-1}},\qquad
	s_{\mathbf{MID}}=1-s_{\mathbf{AS}}-s_{\mathbf{EDGE}}.
	\]
	Also let,
	\[
	\beta:=s_{\mathbf{AS}}+s_{\mathbf{EDGE}}=\frac{k+1}{2^{\,k-1}}\qquad\text{and}\qquad
	A_\ell:=1-(1-s_{\mathbf{MID}})^\ell=1-\beta^\ell.
	\]
	
	\subsection{The MIDDLE-HEAVY rule}
	Under the MIDDLE-HEAVY rule we always select a \textbf{MID} clause if present; otherwise select from \textbf{EDGE} if present; otherwise from \textbf{AS}. 
	By exchangeability of the \(\ell\) candidates, conditional on the selected priority class, the chosen clause is a uniformly random clause from that class. 
	Using the projection in~\eqref{eq_X}, within each class the projected 2-clause type distribution is:
	\[
	\begin{array}{c|ccc}
		\text{class} & \text{type }-- & \text{type }+- & \text{type }++\\\hline
		\mathbf{AS}\ (\vX\in\{0,k\}) & \tfrac12 & 0 & \tfrac12\\
		\mathbf{EDGE}\ (\vX\in\{1,k-1\}) & 0 & \tfrac12 & \tfrac12\\
		\mathbf{MID}\ (2\le \vX\le k-2) & 0 & 0 & 1
	\end{array}
	\]
	The probability that the selected clause lies in \(\mathbf{MID}\) equals \(A_\ell=1-(1-s_{\mathbf{MID}})^\ell\).
	If no \(\mathbf{MID}\) occurs (probability \(\beta^\ell\)), then we select from \(\mathbf{EDGE}\) unless all \(\ell\) are \(\mathbf{AS}\), which has probability \(s_{\mathbf{AS}}^\ell\). Hence
	\[
	\Pr[\text{select }\mathbf{EDGE}]=(1-s_{\mathbf{MID}})^\ell - s_{\mathbf{AS}}^\ell=\beta^\ell - s_{\mathbf{AS}}^\ell,\qquad
	\Pr[\text{select }\mathbf{AS}]=s_{\mathbf{AS}}^\ell.
	\]
	Combining, the exact selection frequencies are
	\begin{equation}\label{eq:p012-sym}
			\begin{aligned}
				p_0^{\mathrm{SYM}} &= \frac12\,s_{\mathbf{AS}}^\ell,\\[2pt]
				p_1^{\mathrm{SYM}} &= \frac12\,(\beta^\ell - s_{\mathbf{AS}}^\ell),\\[2pt]
				p_2^{\mathrm{SYM}} &= A_\ell + \frac12\,(\beta^\ell - s_{\mathbf{AS}}^\ell) + \frac12\,s_{\mathbf{AS}}^\ell
				\;=\;1-\frac12\,\beta^\ell.
			\end{aligned}
	\end{equation}
	Therefore
	\begin{equation}\label{eq:Q-sym}
		Q_{\mathrm{SYM}} \;=\; p_1^{\mathrm{SYM}} + 2\sqrt{p_0^{\mathrm{SYM}}p_2^{\mathrm{SYM}}}
		\;=\;\frac12\,(\beta^\ell - s_{\mathbf{AS}}^\ell)\;+\;2\sqrt{\frac12\,s_{\mathbf{AS}}^\ell\Bigl(1-\frac12\,\beta^\ell\Bigr)}.
	\end{equation}
	By Theorem~\ref{thm:main-symmetric}, the certificate guarantees satisfiability w.h.p.\ whenever
	\(\alpha<1/Q_{\mathrm{SYM}}\).
	\begin{proof}[Proof of \Thm~\ref{thm:main-symmetric}]
		Finally to prove the \Thm~\ref{thm:main-symmetric} using the MIDDLE-HEAVY symmetric rule, we consider few cases.
		\begin{itemize}
			\item For $k=4$, we have $\alpha_{\mathbf{SYM}}(4,5)=1/Q_{\mathrm{SYM}}=18.086\dots$ which is strictly greater than the best known asymptotic random $k$-SAT upper bound $\sim 2^k\log 2$ (for more details refer to Table~\ref{tab:intro-compare}), so $\ell=5$ choices are enough.
			\item 	For $k=5$, numerically we check that $\alpha_{\textbf{SYM}}(5,4)=65.032\dots$ which is again strictly larger than the best known asymptotic random $k$-SAT upper bound $\sim 2^k\log 2$, so $\ell=4$ choices are enough.
			\item For $k\geq 6$, we again numerically check that $\alpha_{\mathbf{SYM}}(k,3)>2^k\log 2$ and moreover the function $\frac{2^k\log 2}{\alpha_{\mathbf{SYM}}(k,\ell)}$ is decreasing in $k$, so $3$ choices suffice for all $k\geq 6$.
		\end{itemize}
	\end{proof}
	\subsection{The Threshold-Symmetric Hybrid rule}
	From Algorithm~\ref{alg_hybrid} the hybrid rule first prefers \textbf{MID} as above; if none are present among the \(\ell\) candidates, it flips an unbiased coin \(b\in\{0,1\}\). If \(b=0\) it selects the clause with \emph{maximum} \(\vX\) (favoring positive clauses); if \(b=1\) it selects the clause with \emph{minimum} \(\vX\) (favoring negative clauses). Thus the only difference from MIDDLE-HEAVY occurs on the event \(\{\text{no MID}\}\), which has probability \(\beta^\ell\).
	
	Conditional on \(\{\text{no \textbf{MID}}\}\), every candidate has \(\vX\in\{0,1,k-1,k\}\), with per-candidate probabilities proportional to \((1,k,k,1)\). Hence under \(\{\text{no \textbf{MID}}\}\) the four values are i.i.d.\ with
	\[
	\Pr(\vX=0)=\frac{1}{2(k+1)},\quad
	\Pr(\vX=1)=\frac{k}{2(k+1)},\quad
	\Pr(\vX=k-1)=\frac{k}{2(k+1)},\quad
	\Pr(\vX=k)=\frac{1}{2(k+1)}.
	\]
	A direct order-statistics calculation yields the conditional type distribution:
	\[
	\begin{array}{c|ccc}
		\text{branch on }\{\text{no MID}\} & \text{type }-- & \text{type }+- & \text{type }++\\\hline
		b=0\ (\max \vX) & (2(k{+}1))^{-\ell} & 2^{-\ell}-(2(k{+}1))^{-\ell} & 1-2^{-\ell}\\
		b=1\ (\min \vX) & 1-\left(1-\tfrac{1}{2(k+1)}\right)^\ell & \left(1-\tfrac{1}{2(k+1)}\right)^\ell-2^{-\ell} & 2^{-\ell}
	\end{array}
	\]
	Averaging over the fair coin \(b\) and re-weighting by \(\beta^\ell=\Pr(\text{no MID})\), while \(\Pr(\text{select MID})=A_\ell\) contributes entirely to type \(++\), we get the exact frequencies:
	\begin{equation}\label{eq:p012-hyb}
			\begin{aligned}
				p_0^{\mathrm{HYB}} &= \frac{\beta^\ell}{2}\left[\,1-\Bigl(1-\frac{1}{2(k+1)}\Bigr)^\ell + \Bigl(\frac{1}{2(k+1)}\Bigr)^\ell\right],\\[4pt]
				p_1^{\mathrm{HYB}} &= \frac{\beta^\ell}{2}\left[\,\Bigl(1-\frac{1}{2(k+1)}\Bigr)^\ell - \Bigl(\frac{1}{2(k+1)}\Bigr)^\ell\right],\\[4pt]
				p_2^{\mathrm{HYB}} &= A_\ell \;+\;\frac{\beta^\ell}{2}.
			\end{aligned}
	\end{equation}
	 Hence,
	\begin{equation}\label{eq:Q-hyb}
		Q_{\mathrm{HYB}} \;=\; p_1^{\mathrm{HYB}} + 2\sqrt{p_0^{\mathrm{HYB}}p_2^{\mathrm{HYB}}}
		\;=\;\frac{\beta^\ell}{2}\left[\,\Bigl(1-\frac{1}{2(k+1)}\Bigr)^\ell - \Bigl(\frac{1}{2(k+1)}\Bigr)^\ell\right]
		+2\sqrt{p_0^{\mathrm{HYB}}\Bigl(A_\ell+\frac{\beta^\ell}{2}\Bigr)}.
	\end{equation}
	Therefore the certificate applies whenever \(\alpha<1/Q_{\mathrm{HYB}}\).
	
	\begin{remark}\label{rem:hyb-conditional}
		In each run the coin $b$ is fixed for all steps. Conditioning on $b$:
		\begin{itemize}
			\item If $b=0$ (maximize $X$ on the fallback), the \emph{coin-conditioned} frequencies are
			\[
			p_0=2^{-k\ell},\qquad
			p_1=2^{-k\ell}\big((k+1)^\ell-1\big),\qquad
			p_2=1-2^{-k\ell}(k+1)^\ell,
			\]
			which are exactly the "max-positives" formulas.
			\item If $b=1$ (minimize $X$), the coin-conditioned frequencies swap $p_0$ and $p_2$, with the same $p_1$.
		\end{itemize}
		Since the spectral parameter $Q=p_1+2\sqrt{p_0p_2}$ is symmetric in $p_0,p_2$,
		both branches yield the \emph{same} $Q$ (call it $Q_{\max}$). Thus one may
		equivalently apply the certificate \emph{conditional on $b$}, obtaining the
		threshold $\alpha<1/Q_{\max}$. The coin-\emph{averaged} formulas
		\eqref{eq:p012-hyb} are the unconditional selection frequencies of the hybrid
		process; using them also certifies $\alpha<1/Q_{\HYB}$, and conditioning shows
		$Q_{\HYB}$ and $Q_{\max}$ lead to the same (or stronger) sufficient condition.
	\end{remark}
	\begin{proof}[Proof of \Thm~\ref{thm:hybrid}]
		Finally to prove the \Thm~\ref{thm:hybrid} using the Threshold-Symmetric Hybrid rule, we consider few cases.
		\begin{itemize}
			\item For $k=4$, we have $\alpha_{\mathbf{SYM}}(4,4)=1/Q_{\mathrm{SYM}}=11.935\dots$ which is greater than the best known asymptotic random $k$-SAT upper bound $\sim 2^k\log 2$, so $\ell=4$ choices are enough.
			Note that, here we used the unbiased Hybrid rule where flipping a coin will give outcome either $0$ or $1$ with equal probability. 
			Of course for the biased hybrid rule one can beat the known asymptotic upper bound for the satisfiability threshold with $\ell=3$ choices (for specific threshold values refer to Table~\ref{tab:intro-compare}).
			\item For $k\geq 5$, we again numerically check that $\alpha_{\mathbf{SYM}}(k,3)>2^k\log 2$ and moreover the function $\frac{2^k\log 2}{\alpha_{\mathbf{HYB}}(k,\ell)}$ is decreasing in $k$, so $3$ choices suffice for all $k\geq 5$.
		\end{itemize}
	\end{proof}
		
	\section{Discussions and Open Problems}\label{sec:conclusions}
    The central idea of this work is to provide a purely \emph{assignment--symmetric}, online, topology--oblivious rule which can push the satisfiability threshold of random $k$-SAT strictly to the right for minimal~$\ell$ candidate choices.
    Our analysis establishes a certificate framework for the satisfiability of the semi-random Achlioptas process under sign-profile based candidate selection rules using a fixed $2$-SAT projection. 
    The core object is to analyze the two-type implication exploration with mean matrix
    \[
    M(\alpha)=\alpha
    \begin{pmatrix}
    	p_1 & 2p_0\\[2pt]
    	2p_2 & p_1
    \end{pmatrix},\qquad
    Q:=p_1+2\sqrt{p_0p_2},
    \]
    and the sufficiency condition $\alpha<1/Q$. 
    For both the rules symmetric and hybrid, we computed the clause type frequencies $(p_0,p_1,p_2)$ and hence obtaining explicit thresholds $\alpha=1/Q$.  
    Although the $1/Q$ bound is conventional, improving it without leaving a self-contained short certificate is a challenging task.
    Moreover our Threshold-Symmetric Hybrid algorithm flips a fair coin at the beginning of the process once.
    Conditioned on the outcome of the result, the process is either max-positives or min-positives on the non-MID fallback case. 
    Although the frequencies $(p_0,p_2)$ differ between the two branches max-positives and min-positives respectively, the certificate parameter
    \[
    Q=p_1+2\sqrt{p_0p_2}
    \]
    is invariant because it is symmetric in $p_0,p_2$. Consequently, the satisfiability threshold is the same in either of the branch and strictly larger than the threshold computed from the unbiased (coin-averaged) frequencies.
    But this is the first symmetric based rule as compared to the previous works \cite{DDHM,Perkins,SV} for the semi-random model in any random constraint satisfaction problems.
     
    We conclude with several open problems that we believe are both effective and accessible.
    
   \subsection*{(i).}
   Our results give \emph{lower bounds} on the threshold (SAT certificates). 
   There is an open question regarding the development of the  \emph{upper bounds} on the UNSAT certificates specific to these Achlioptas rules:
   \begin{itemize}
   	\item Show that for $\alpha> \alpha^\star(k,\ell)$ the formula is unsatisfiable w.h.p. under the MIDDLE-HEAVY or the Threshold-Symmetric Hybrid rule.
   	\item As a candidate approach one can think in the direction of unit-clause process with drift analysis, density evolution for pure-literal elimination or the emergence of a giant contradictory Strongly Connected Component (SCC) in the implication digraph.
   \end{itemize}
   
   \subsection*{(ii).}
   Can one design a certificate that adapts to the evolving formula, e.g., via a simplified process or a hybrid exploration that tracks unit implications as they are created?
   Also, our branching bound ignores correlations introduced by selecting a clause from $\ell$ candidates. 
   Can we exploit these correlations to certify a strictly smaller $Q$?
   
   \subsection*{(iii).}
   In this paper, we use the online Achlioptas process for calculating the thresholds in random $k$-SAT.
   One can think of the improvement in the threshold values achievable by assignment-symmetric offline version relative to the online process.
   
   \subsection*{(iv).}
   In this paper we use the static $2$-SAT projection rule.
   One interesting direction can be to extend the analysis of the projection idea to NAE-SAT or XOR-SAT and any other general $2$-CSPs with two spin types.
   
  And finally, many other works on the phase transition of Achlioptas random graph process either by accelerating or delaying the birth of giant component in several graph properties suggest many directions for future works.

\end{document}